\documentclass[11pt]{article}

\usepackage{times}
\usepackage{fullpage}
\usepackage{microtype,pdfsync}
\usepackage{amsmath,amsfonts,amssymb,amsthm}
\usepackage{url}
\usepackage{enumitem}
\usepackage{xspace}
\usepackage{hyperref}
\usepackage{algorithm}
\usepackage[noend]{algorithmic}

\newcommand{\Q}{\ensuremath{\mathbb{Q}}}
\newcommand{\R}{\ensuremath{\mathbb{R}}}

\newcommand{\Z}{\ensuremath{\mathbb{Z}}}

\newcommand{\pr}[2]{\left\langle#1, #2\right\rangle}

\newcommand{\abs}[1]{\lvert{#1}\rvert}
\newcommand{\absfit}[1]{\left\lvert{#1}\right\rvert}
\newcommand{\set}[1]{\{{#1}\}}

\newcommand{\floor}[1]{\lfloor{#1}\rfloor}

\newcommand{\length}[1]{\lVert{#1}\rVert}

\newcommand{\veca}{\ensuremath{\mathbf{a}}}
\newcommand{\vecb}{\ensuremath{\mathbf{b}}}
\newcommand{\vecc}{\ensuremath{\mathbf{c}}}

\newcommand{\vecq}{\ensuremath{\mathbf{q}}}

\newcommand{\vect}{\ensuremath{\mathbf{t}}}

\newcommand{\vecv}{\ensuremath{\mathbf{v}}}
\newcommand{\vecw}{\ensuremath{\mathbf{w}}}
\newcommand{\vecx}{\ensuremath{\mathbf{x}}}
\newcommand{\vecy}{\ensuremath{\mathbf{y}}}
\newcommand{\vecz}{\ensuremath{\mathbf{z}}}
\newcommand{\veczero}{\ensuremath{\mathbf{0}}}

\theoremstyle{plain}            
\newtheorem{theorem}{Theorem}[section]
\newtheorem{lemma}[theorem]{Lemma}

\theoremstyle{definition}       
\newtheorem{definition}[theorem]{Definition}

\theoremstyle{remark}           

\numberwithin{equation}{section}

\DeclareMathOperator{\poly}{poly}

\DeclareMathOperator*{\E}{\mathbb{E}}

\newcommand{\lat}{\mathcal{L}}

\DeclareMathOperator{\vol}{vol}

\newif\ifnotes\notesfalse

\ifnotes

\newcommand{\dnote}[1]{{\bf (Daniel:} {#1}{\bf ) }}
\newcommand{\gnote}[1]{{\bf (Gabor:} {#1}{\bf )}}

\else

\newcommand{\dnote}[1]{}
\newcommand{\gnote}[1]{}

\fi

\renewcommand{\epsilon}{\varepsilon}
\newcommand{\eps}{\epsilon}

\DeclareMathOperator{\SVP}{SVP}
\DeclareMathOperator{\CVP}{CVP}

\DeclareMathOperator{\argmin}{arg\,min}
\def\minuszero{\setminus \set{\veczero}}

\makeatletter
\def\imod#1{\allowbreak\mkern10mu({\operator@font mod}\,\,#1)}
\makeatother

\title{Lattice Sparsification and the Approximate Closest Vector Problem}

\author{
  Daniel Dadush\thanks{Georgia Tech. Atlanta, GA, USA. Email: \texttt{dndadush@gatech.edu}}
  \and
  G\'abor Kun\thanks{Courant Institute for Mathematical Sciences, NY, USA. Email: \texttt{kungabor@cs.elte.hu}}
}

\begin{document}

\maketitle
\thispagestyle{empty}

\begin{abstract}
  We give a deterministic algorithm for solving the $(1+\eps)$ approximate Closest Vector Problem (CVP) on any n dimensional lattice and any norm in
$2^{O(n)}(1+1/\eps)^n$ time and $2^n\poly(n)$ space. Our algorithm builds on the lattice point enumeration techniques of Micciancio and Voulgaris (STOC 2010)
and Dadush, Peikert and Vempala (FOCS 2011), and gives an elegant, deterministic alternative to the ``AKS Sieve'' based algorithms for $(1+\eps)$-CVP
(Ajtai, Kumar, and Sivakumar; STOC 2001 and CCC 2002).  Furthermore, assuming the existence of a $\poly(n)$-space and $2^{O(n)}$ time algorithm for
exact CVP in the $l_2$ norm, the space complexity of our algorithm can be reduced to polynomial.

Our main technical contribution is a method for ``sparsifying'' any input lattice while approximately maintaining its metric structure. To this end, we
employ the idea of random sublattice restrictions, which was first employed by Khot (FOCS 2003) for the purpose of proving hardness for Shortest
Vector Problem (SVP) under $l_p$ norms.


\end{abstract}

\newpage
\pagenumbering{arabic} 

\section{Introduction}
\label{sec:lat-introduction}

An $n$-dimensional lattice $\lat$ is $\set{\sum_{i=1}^n z_i \vecb_i: z_i \in \Z, i \in [n]}$ for some basis $\vecb_1,\dots,\vecb_n$ of $\R^n$. Given a
lattice $\lat$ and norm $\|\cdot\|$ in $\R^n$, the Shortest Vector Problem (SVP) is to find a shortest \emph{nonzero} $\vecv \in \lat$ under
$\|\cdot\|$. Given an additional target $t \in \R^n$, the Closest Vector Problem (CVP) -- the inhomogenous analog of SVP -- is to find a closest $\vecv
\in \lat$ to $\vect$. Here, one often works with the $\ell_2$ norm and other $\ell_{p}$ norms, or most generally, 
with norms (possibly asymmetric) induced by a convex body $K$ containing $0$ in its interior, defined by $\length{\vecx}_{K} = \inf \set{ s \geq 0 : \vecx \in sK }$.  

The SVP and CVP on lattices are central algorithmic problems in the geometry of numbers, with applications to Integer
Programming~\cite{lenstra83:_integ_progr_with_fixed_number_of_variab}, factoring polynomials over the rationals~\cite{lenstra82:_factor},
cryptoanalysis (e.g.,~\cite{odlyzko90:_rise_and_fall_of_knaps_crypt,DBLP:journals/joc/JouxS98,DBLP:conf/calc/NguyenS01}), and much more. For different
applications, one must often consider lattice problems expressed under a variety of norms. Decoding signals over a Gaussian channel is expressed as a
CVP under $\ell_2$~\cite{journal/TransIT/ViterboB99}, computing simultaneous diophantine approximations is generally expressed as an SVP under
$\ell_\infty$~\cite{journal/combinatorica/FrankT87}, Schnorr reduced factoring (under some unproven number theoretic assumptions) to an
SVP under the $\ell_1$ norm~\cite{conf/eurocrypt/Schnorr91}, the Frobenius problem can be expressed as a lattice problem under an asymmetric
simplicial norm~\cite{journal/combinatorica/Kannan92}, the Integer Programming problem reduces to lattice problems under general
norms~\cite{kannan87:_minkow_convex_body_theor_and_integ_progr,conf/focs/svp/DPV11}, etc. 

Much is known about the computational complexity of SVP and CVP, in both their exact and approximation versions. On the negative side, SVP is NP-hard
(in $\ell_{2}$, under randomized reductions) to solve exactly, or even to approximate to within any constant
factor~\cite{DBLP:conf/stoc/Ajtai98,DBLP:journals/jcss/CaiN99,DBLP:journals/siamcomp/Micciancio00,DBLP:journals/jacm/Khot05}. Many more hardness
results are known for other $\ell_{p}$ norms and under stronger complexity assumptions than P $\neq$ NP (see,
e.g.,~\cite{emde81:_anoth_np,DBLP:journals/tcs/Dinur02,DBLP:conf/stoc/RegevR06,DBLP:conf/stoc/HavivR07}). CVP is NP-hard to approximate to within
$n^{c/\log \log n}$ factors for some constant $c > 0$~\cite{DBLP:journals/jcss/AroraBSS97,DBLP:journals/combinatorica/DinurKRS03,DBLP:journals/tcs/Dinur02}, 
where $n$ is the dimension of the lattice. Therefore, we do not expect to solve (or even closely approximate) these problems efficiently in high dimensions.  
Still, algorithms providing weak approximations or having super-polynomial running times are the foundations for the many applications mentioned above.

Though the applications are often expressed using a variety of norms, the majority of the algorithmic work on SVP and CVP over the last quarter
century has focused on the important case of the $\ell_2$ norm. While there has been both tremendous practical and theoretical progress for $\ell_2$
based solvers, progress on more general norms has been much slower (we overview this history below). Illustrative of this, for most of the problems
mentioned above, the solution strategy has almost invariably been to approximate the problem via a reduction to $\ell_2$. In many cases, the desired
computational problem requires only a ``coarse'' approximate solution to the underlying lattice problem (e.g. where a $\poly(n)$ or even $2^{O(n)}$
factor approximation suffices), in which case approximation by $\ell_2$ is often sufficient. In some cases however, the errors induced by the
$\ell_2$ approximation can result in a substantial increase in worst case running time or yield unusable results. As an example, with respect to the
Integer Programming Problem (IP), in a sequence of works Dadush, Peikert and Vempala~\cite{conf/focs/svp/DPV11,thesis/D12} worked directly with norms induced by the
continuous relaxation -- avoiding direct ellipsoidal approximations -- to reduce the complexity of solving an $n$-variable IP from $2^{O(n)} n^{2n}$
(previous best using $\ell_2$ techniques~\cite{arxiv/HildebrandK10}) to $2^{O(n)} n^n$. From these considerations we see that the problem of
developing effective algorithms for solving the SVP and CVP under general norms is motivated. 

The algorithmic history of the SVP and CVP is long and rich. We relate the broad outlines here, highlighting the pertinent developments for general
norms, and refer the reader to the following references~\cite{MGBook, conf/cc/HPS11} for a more complete accounting. There are three main classes of
methods for solving lattice problems: basis reduction, randomized sieving, and Voronoi cell based search. 

{\bf Basis reduction} combines both local search on lattice bases and lattice point enumeration. The celebrated LLL basis reduction algorithm~\cite{lenstra82:_factor} and further extensions~\cite{DBLP:journals/combinatorica/Babai86,DBLP:journals/tcs/Schnorr87} give $2^{n /
\text{polylog}(n)}$ approximations to SVP and CVP under $\ell_{2}$ in $\poly(n)$ time.  General norm variants of basis reduction are explored
in~\cite{lovasz92,manuscript/KR95} and give similar approximation guarantees for SVP (though not CVP) as the $\ell_2$ versions. However, bounds on the time complexity were only proved for fixed dimension (when the running time is polynomial). For exact SVP and CVP in the $\ell_{2}$ norm, Kannan's algorithm and its subsequent improvements~\cite{kannan87:_minkow_convex_body_theor_and_integ_progr,journal/tcs/Helfrich85,conf/crypto/HanrotS07} use basis reduction techniques to deterministically compute solutions in $2^{O(n \log n)}$ time and $\poly(n)$ space. 

This performance remained essentially unchallenged until the breakthrough {\bf randomized ``sieving''} algorithm of Ajtai, Kumar, and
Sivakumar~\cite{DBLP:conf/stoc/AjtaiKS01}, which gave a $2^{O(n)}$-time and -space randomized algorithm for exact SVP under $\ell_2$.  The
randomized sieving approach consists of sampling an exponential number of ``perturbed'' lattice points, and then iteratively clustering and
combining them to give shorter and shorter lattice points.  Subsequently, the randomized sieve was greatly extended to yield solutions for
more general norms and for the more general problem of $(1+\eps)$-CVP. For exact SVP, the randomized sieve was extended (in the same time
complexity) to $\ell_p$ norms~\cite{DBLP:conf/icalp/BlomerN07}, arbitrary symmetric norms~\cite{DBLP:conf/fsttcs/ArvindJ08}, and to
``near-symmetric''~\footnote{An asymmetric norm with unit ball $K \subseteq \R^n$ is near-symmetric if $\vol_n(K) \leq 2^{O(n)} \vol_n(K
\cap - K)$.} norms\cite{conf/latin/cvp/D12}. For CVP, the randomized sieve was further used to give a $(\frac{1}{\eps})^n$-time and -space
algorithm for $(1+\eps)$-CVP under the $\ell_2$ norm~\cite{DBLP:conf/coco/AjtaiKS02,DBLP:conf/icalp/BlomerN07}, $\ell_p$
norms~\cite{DBLP:conf/icalp/BlomerN07}, and near-symmetric norms~\cite{conf/latin/cvp/D12}. We remark that near-symmetric norms appear
naturally in the context of Integer Programming: the problem of finding a lattice point near the ``center'' of the continuous relaxation
(which need not be symmetric) can be directly expressed as a CVP under a near-symmetric norm~\cite{conf/latin/cvp/D12}. Lastly, for the
specific case of $\ell_\infty$, Eisenbrand, H\"{a}hnle and Niemeier~\cite{conf/socg/EisenbrandHN11} show that $(1+\eps)$-CVP under
$\ell_\infty$ can be solved using $O(\ln \frac{1}{\eps})^n$ calls to any $2$-approximate solver via an elegant cube covering technique. It
is worth noting that AKS sieve based algorithms are \emph{Monte Carlo}: while they output correct solutions (i.e.~a shortest or closest
vectors) with high probability, the correctness is not guaranteed.

In a major breakthrough, Micciancio and Voulgaris~\cite{DBLP:conf/stoc/MicciancioV10} gave a \emph{deterministic} $2^{O(n)}$-time and -space algorithm
for exact SVP and CVP under the $\ell_{2}$ norm using the {\bf Voronoi cell} of a lattice. The Voronoi cell, the symmetric polytope consisting of all points in
space closer to the origin (under $\ell_2$) than any other lattice point, is represented algorithmically here by $O(2^n)$ lattice points
corresponding to the facets of the Voronoi cell (known as Voronoi relevant vectors). The relevant vectors form an ``extended basis'' for the lattice
which Micciancio and Voulgaris (MV) use to efficiently guide closest lattice point search. Though it is tempting to try and directly extend the MV
techniques to other norms this appears to be quite challenging. A major difficulty is that for general norms the Voronoi cell need not be convex, and
furthermore no good bounds are known for the number of relevant vectors. In a subsequent work however, Dadush, 
Peikert and
Vempala~\cite{conf/focs/svp/DPV11} showed that MV lattice point search techniques can, in a qualified sense, be extended to general norms (in fact, to
general convex bodies) via a direct reduction to $\ell_2$. Combining a technique for constructing ``efficient'' ellipsoid coverings -- using the
M-Ellipsoid concept from convex geometry -- together with Voronoi cell based search, they showed that the lattice points inside a convex body can be
computed in time proportional to the maximum number of lattice points the body can contain in any translation. With some further
improvements~\cite{arxiv/m-ell/DV12,thesis/D12}, the DPV lattice point enumeration technique was used to give the first deterministic $2^{O(n)}$-time
and -space algorithms for SVP and Bounded Distance Decoding (BDD)~\footnote{BDD is CVP when the distance to the target is guaranteed to be at most
some factor times the minimum distance of the lattice.} under near-symmetric norms. 

Despite all the recent progress, the only algorithms currently available for solving $(1+\eps)$-CVP under non-euclidean norms remain the AKS sieve
based approaches. In this light, a main open problem from~\cite{conf/focs/svp/DPV11} was to understand whether the DPV lattice point enumeration
approach could be extended to work for $(1+\eps)$-CVP under general norms.
 
%
%
%

\subsection{Results and Techniques}
\label{sec:lat-results}

Our main result is as follows:

\begin{theorem}[Approximate CVP in any norm, informal]
\label{thm:approx-cvp-deterministic-informal}
There is a deterministic algorithm that, given any near-symmetric norm $\|\cdot\|_K$, $n$ dimensional lattice $\lat$, target $\vecx \in
\R^n$, and $0 < \eps \leq 1$, computes $\vecy \in \lat$, a $(1+\eps)$-approximate minimizer to $\|\vecy-\vecx\|_K$, in $(1+\frac{1}{\eps})^n
\cdot 2^{O(n)}$ time and $\tilde{O}(2^{n})$ space.
\end{theorem}

In the above theorem we extend the DPV lattice point enumeration techniques and give the first deterministic alternative to the AKS randomized sieving
approach. Compared to AKS, our approach also achieves a better dependence on $\eps$, $2^{O(n)}(1+\frac{1}{\eps})^n$ instead of
$2^{O(n)}(1+\frac{1}{\eps})^{2n}$, and utilizes significantly less space, $\tilde{O}(2^n)$ compared to $2^{O(n)}(1+\frac{1}{\eps})^n$. Additionally, as we
will discuss below, continued progress on exact CVP under $\ell_2$ could further reduce the space usage of the algorithm. We note however that the
$2^{O(n)}$ factors in the running time are currently much larger than in AKS, though little effort has been spent in trying to compute or optimize them. 
To explain our approach, we first present the main DPV enumeration algorithm in its most recent formulation~\cite{thesis/D12}. 

\begin{theorem}[Enumeration in Convex Bodies, informal]
\label{thm:enumeration-informal}
There is a deterministic algorithm that, given an $n$-dimensional convex body $K$ and lattice $\lat$, enumerates the elements of $K \cap \lat$ in time
$2^{O(n)} G(K,\lat)$ using $\tilde{O}(2^n)$ space, where $G(K,\lat) = \max_{\vecx \in \R^n} |(K+\vecx) \cap \lat|$. Furthermore, given an algorithm that
solves exact CVP under $\ell_2$ in $T(n)$ time and $S(n)$ space, $K \cap \lat$ can be enumerated in $2^{O(n)} T(n) G(K,\lat)$ time using $S(n) +
\poly(n)$ space.   
\end{theorem}

The main idea for the above algorithm is to first compute a covering of $K$ by $2^{O(n)}$ translates of an $M$-ellipsoid $E$ of $K$~\footnote{An
M-Ellipsoid $E$ of $K$ satisfies that $2^{O(n)}$ translates of $E$ suffice to cover $K$ and vice versa.}, and to use the MV enumeration techniques to
compute the lattice points inside each translate of $E$. In its first incarnation~\cite{conf/focs/svp/DPV11}, the above algorithm was randomized --
here randomization was needed to construct the M-Ellipsoid -- and had space complexity dependent on $G(K,\lat)$. In~\cite{arxiv/m-ell/DV12}, a
deterministic M-Ellipsoid construction was presented yielding a completely deterministic enumerator. Lastly in~\cite{thesis/D12}, the space usage
was decoupled from $G(K,\lat)$ and a direct reduction from lattice point enumeration to exact CVP under $\ell_2$ was presented. 

The above lattice point enumerator will form the core of our $(1+\eps)$-CVP algorithm. As we will see from the algorithm's analysis, its space usage
will only be an additive polynomial factor larger than the space required for the enumeration. Therefore, if one could develop an exact CVP solver under
$\ell_2$ which runs in $2^{O(n)}$ time and $\poly(n)$ space, then the space usage of our $(1+\eps)$-CVP can be reduced to $\poly(n)$ in the same time
complexity. The possibility of such a solver is discussed in~\cite{DBLP:conf/stoc/MicciancioV10} and developing it remains an important open problem.
We remark that by plugging in Kannan's algorithm for CVP under $\ell_2$, we do indeed get a $\poly(n)$ space $(1+\eps)$-CVP solver, though at the cost
of an $n^{n/2}$ factor increase in running time.

Using the above enumerator as a blackbox, we now present the approach taken in~\cite{conf/focs/svp/DPV11} to solve CVP and explain the main problem
that arises. Given the target $\vect \in \R^n$, their algorithm first computes an initial coarse underestimate $d_0$ of the distance of $\vect$ to
$\lat$ under $\|\cdot\|_K$ (using LLL for example). For the next step, they use the lattice point enumerator to successively compute the sets $(\vect
+ 2^id_0K) \cap \lat$ (i.e.~all lattice points at distance at most $2^id_0$ from $\vect$), $i \geq 0$, until a lattice point is found. Finally, the
closest vector to $\vect$ in the final enumerated set is returned. 

From the description, it is relatively straightforward to show that the complexity of the algorithm is essentially $G(dK,\lat)$, where $d$ is the
distance of $\vect$ to $\lat$. The main problem with this approach is that, in general, one cannot apriori bound $G(dK,\lat)$; even
in $2$ dimension this quantity can be made arbitrarily large. The only generic setting where such a bound is indeed available is when the distance
$d$ of the target is bounded by $\alpha \lambda$, where $\lambda$ is the length of the shortest non-zero vector under $\|\cdot\|_K$. In this
situation, we can bound $G(dK,\lat)$ by $2^{O(n)}(1+\alpha)^n$. We remark that solving CVP with this type of guarantee corresponds to the Bounded
Distance Problem problem in the literature, and by a standard reduction can be used to solve SVP in general norms as well~\cite{DBLP:journals/ipl/GoldreichMSS99}. 

To circumvent the above problem, we propose the following simple solution.  Instead of solving the CVP on the original lattice $\lat$, we attempt to
solve it on a sparser sublattice $\lat' \subseteq \lat$, where the distance of $\vect$ to $\lat'$ is not much larger than its distance to $\lat$ (we
settle for an approximate solution here) and where the maximum number of lattice points at the new target distance is appropriately bounded. Our main
technical contribution is to show the existence of such ``lattice sparsifiers'' and give a deterministic algorithm to compute them:

\begin{theorem}[Lattice Sparsifier, informal]
\label{thm:lattice-sparsifier-informal}
There is a deterministic algorithm that, given any near-symmetric norm $\|\cdot\|_K$, $n$ dimensional lattice $\lat$, and distance $t \geq 0$,
computes a sublattice $\lat' \subseteq \lat$ in deterministic $2^{O(n)}$ time and $\tilde{O}(2^n)$ space satisfying: (1) the distance from $\lat'$ to any
point in $\R^n$ is at most its distance to $\lat$ plus an additive $t$, (2) the number of points in $\lat'$ at distance $t$ is at most $2^{O(n)}$. 
\end{theorem}

To solve $(1+\eps)$-CVP using the above lattice sparsifier is straightforward. We simply compute a sparsifier $\lat'$ for $\lat$ under $\|\cdot\|_K$ with $t = \eps d_K(t, \lat)$
(the distance from $\vect$ to $\lat$)
, and then solve the exact CVP on $\lat'$ using the DPV algorithm. By the guarantees on the sparsifier, $\lat'$
contains a point at distance at most $d + \eps d = (1+\eps) d$, and using a simple packing argument (see Lemma \ref{lem:gkl-bnds}) we can show that
\[
G((1+\eps)d,\lat') = 2^{O(n)}(1+\frac{1}{\eps})^n G(\eps d, \lat') = 2^{O(n)}(1+\frac{1}{\eps})^n \text{.}
\]
Here we note that the correctness of the output follows from the distance preserving properties of $\lat'$, and the desired runtime follows from the
above bound on $G((1+\eps)d,\lat')$.

To prove the existence of lattice sparsifier's we make use of random sublattice restrictions, a tool first employed by
Khot~\cite{DBLP:journals/jcss/Khot06,DBLP:journals/jacm/Khot05} for the purpose of proving hardness of SVP. More precisely, we show that
with constant probability the restriction of $\lat$ by a random modular form (for an appropriately chosen modulus) yields the desired
sparsifier. We remark that our use of sublattice restrictions is somewhat more refined than
in~\cite{DBLP:journals/jcss/Khot06,DBLP:journals/jacm/Khot05}. In Khot's setting, the random sublattice is calibrated to remove all short
vectors on a NO instance, and to keep at least one short vector for a YES instance. In our setting, we somehow need both properties
simultaneously for the \emph{same} lattice, i.e.~we want to remove many short vectors to guarantee reasonable enumeration complexity, while
at the same time keeping enough vectors so that the original lattice lies ``close'' to the sublattice. As a final difference, we show that
our construction can be derandomized in $2^{O(n)}$ time, yielding a completely deterministic algorithm. 

%
\paragraph{Organization.} In section~\ref{sec:apx-cvp}, we provide the exact reduction from $(1+\eps)$-CVP to lattice sparsification, formalizing Theorem~\ref{thm:approx-cvp-deterministic-informal}. In section~\ref{sec:lat-spar-rand}, we prove the existence of lattice sparsifiers using the
probabilistic method.  In section~\ref{sec:lat-spar-det}, we give the derandomized lattice sparsifier construction, formalizing
Theorem~\ref{thm:lattice-sparsifier-informal}. Lastly, in section~\ref{sec:future-directions}, we discuss futher applications and future directions.

\section{Preliminaries}
\label{sec:preliminaries}

\paragraph{Convexity and Norms.} For sets $A,B \subseteq \R^n$, let $A+B = \set{a+b: a \in A, b \in B}$ denote their Minkowski sum. $B_2^n$ denotes the $n$-dimensional euclidean unit ball in $\R^n$. A convex body $K \subseteq \R^n$ is a full dimensional
compact, convex set. A convex body $K$ is $(\veca_0,r,R)$-centered if $\veca_0+rB_2^n \subseteq K \subseteq \veca_0+RB_2^n$. For a convex
body $K \subseteq \R^n$ containing $\veczero$ in its interior, we define the (possibly asymmetric) norm $\|\cdot\|_K$ induced by $K$ as
$\|\vecx\|_K = \inf \set{s \geq 0: \vecx \in sK}$. For a $(\veczero,r,R)$-centered convex body $K$, we note that $\frac{1}{R}
\|\vecx\|_2 \leq \|\vecx\|_K \leq \frac{1}{r} \|\vecx\|_2$. 

If $K$ is symmetric ($K=-K$), then $\|\cdot\|_K$ is also symmetric ($\|\vecx\|_K = \|-\vecx\|_K$), and hence defines a regular norm on
$\R^n$.  The convex body $K$ ($\|\cdot\|_K$) is $\gamma$-symmetric for $\gamma \in (0,1]$, if $\vol_n(K \cap -K) \geq \gamma^n \vol_n(K)$. $K$ is
near-symmetric if it is $\Omega(1)$-symmetric.

\paragraph{Computational Model.} The convex bodies and norms will be presented to our algorithms via weak membership and distance oracles. For $\eps
\geq 0$ and $K \subseteq \R^n$ a convex body, we define $K^{\eps} = K + \eps B_2^n$ and $K^{-\eps} = \set{\vecx \in K: \vecx + \eps B_2^n \subseteq
K}$.  A \emph{weak membership oracle} $O_K$ for $K$ is a function which takes as input a point $\vecx \in \Q^n$ and real $\eps > 0$, and returns
$O_K(\vecx,\eps) = 1$ if $\vecx \in K^{-\eps}$, $0$ if $\vecx \notin K^{\eps}$, and either $0$ or $1$ if $\vecx \in K^{\eps} \setminus K^{-\eps}$. A
\emph{weak distance oracle} $D_{K, \cdot}$ for $K$ is a function that takes as input a point $\vecx \in \Q^n$ and $\eps > 0$, and returns a rational
number satisfying $\absfit{D_{K, \eps}(\vecx)-\length{\vecx}_K} \leq \eps \min \set{1, \length{\vecx}_K}$. The runtimes of our algorithms will be
measured by the number of oracle calls and arithmetic operations. For simplicity, we use the notation $\poly(\cdot)$ to denote a polynomial factor in
all the relevant input parameters (dimension, encoding length of basis, etc.).   

\paragraph{Lattices.} An $n$-dimensional lattice $\lat \subset \R^{n}$ is a discrete subgroup of $\R^n$; $\lat$ can be expressed as $B\Z^n$,
where $B \in \R^{n \times n}$ is a non-singular matrix, which we refer to as a basis for $\lat$. The dual lattice of $\lat$ is $\lat^* = \set{\vecy \in \R^n: \forall \vecx \in \lat \text{   } \pr{\vecx}{\vecy} \in \Z}$, 
which can be generated by the basis $B^{-T}$ (inverse transpose).  


We define the length of the shortest non-zero vector of $\lat$ under $\|\cdot\|_K$ by $\lambda_1(K,\lat) = \min_{\vecy \in \lat \setminus
\set{\veczero}} \|\vecy\|_K$. We let $\SVP(K,\lat) = \argmin_{\vecz \in \lat \minuszero} \|\vecz\|_K$ denote the set of shortest non-zero
vectors of $\lat$ under $\|\cdot\|$.  For $\vecx \in \R^n$, define the distance of $\vecx$ to $\lat$ under $\|\cdot\|_K$ by $d_K(\lat,\vecx)
= \min_{\vecy \in \lat} \|\vecy-\vecx\|_K$. We let $\CVP(K,\lat,\vecx) = \argmin_{\vecy \in \lat} \|\vecy-\vecx\|_K$ denote the set of
closest vectors to $\vecx$ in $\lat$ under $\|\cdot\|_K$.

For a lattice $\lat$ and convex body $K$ in $\R^n$, let $G(K,\lat)$ be the largest number of lattice points contained in any translate of
$K$, that is $G(K,\lat) = \max_{\vecx \in \R^{n}} \abs{(K+\vecx) \cap \lat}$. We will need the following bounds on $G(K,\lat)$
from~\cite{thesis/D12} (we include a proof in the appendix for completeness).

\begin{lemma}
\label{lem:gkl-bnds}
Let $K \subseteq \R^n$ denote a $\gamma$-symmetric convex body and let $\lat$ denote an $n$-dimensional lattice. Then for $d > 0$ we have that
\[
   G(dK,\lat) \leq \gamma^{-n} \left(1 + \frac{2d}{\lambda_1(K \cap -K,\lat)}\right)^n \quad \text{ and } \quad G(dK,\lat) \leq \gamma^{-n}(2d+1)^n \cdot |(K \cap -K) \cap \lat| \text{.}
\]
\end{lemma}

\paragraph{Algorithms.} We will need the following lattice point enumeration algorithm from~\cite{conf/focs/svp/DPV11,thesis/D12}.

\begin{theorem}[Algorithm Lattice-Enum($K,\lat,\eps$)]
\label{thm:body-lat-enum-informal}
Let $K \subseteq \R^n$ be a $(\veca_0,r,R)$-centered convex body given by weak membership oracle $O_K$, let $\lat \subseteq \R^n$ be an
$n$-dimensional lattice with basis $B \in Q^{n \times n}$ and let $\eps > 0$. Then there is a deterministic algorithm that on inputs
$K,\lat,\eps$ outputs a set $S$ (one element at a time) satisfying 
\[
K \cap \lat \subseteq S \subseteq (K+\eps B_2^n) \cap \lat
\] 
in $G(K,\lat) \cdot 2^{O(n)} \cdot \poly(\cdot)$ time using $2^{n}\poly(\cdot)$ space.
\end{theorem}

We will require the following SVP solver from~\cite{conf/focs/svp/DPV11,thesis/D12}.

\begin{theorem}[Algorithm Shortest-Vectors($K,\lat,\eps$)]
\label{thm:svp-deterministic-informal}
Let $K \subseteq \R^n$ be a $(\veca_0,r,R)$-centered symmetric convex body given by weak membership oracle $O_K$, and let $\lat \subseteq \R^n$ be an
$n$-dimensional lattice with basis $B \in Q^{n \times n}$, and let $\eps > 0$. Let $\lambda_1 = \lambda_1(K,\lat)$. Then there is an
algorithm that on inputs $K,\lat,\eps$ outputs a set $S \subseteq \lat$ satisfying
\begin{equation}
\label{eq:svp-guarantee}
   \SVP(K,\lat) \subseteq S \subseteq \set{\vecy \in \lat \setminus \set{\veczero}: \|\vecy\|_K \leq \lambda_1 + \eps \min \set{1, \lambda_1}}
\end{equation} 
in deterministic $2^{O(n)}\poly(\cdot)$ time and $2^{n}\poly(\cdot)$ space.
\end{theorem}



\section{CVP via Lattice Sparsification}
\label{sec:apx-cvp}

%

To start, we give a precise definition of the lattice sparsifier.

\begin{definition}[Lattice Sparsifier]
\label{def:sparsifier} Let $K \subseteq \R^n$ be a $\gamma$-symmetric convex body, $\lat$ be an $n$-dimensional lattice and $t \geq 0$. A $(K,t)$ sparsifier for $\lat$ is a sublattice $\lat' \subseteq \lat$ satisfying
\begin{enumerate}
\item $\forall \vecx \in \R^n$, $d_K(\lat',\vecx) \leq d_K(\lat, \vecx) + t$
\item $G(tK,\lat) = 2^{O(n)}\gamma^{-n}$
\end{enumerate}
\end{definition} 

%

%
%

The following theorem represents the formalization of our lattice sparsifier construction. 

\begin{theorem}[Algorithm Lattice-Sparsifier] 
\label{thm:lattice-sparsifier}
Let $K \subseteq \R^n$ be a $(\veczero,r,R)$-centered and $\gamma$-symmetric convex body specified by a weak membership
oracle $O_K$, and let $\lat$ denote an $n$ dimensional lattice with a basis $B \in \Q^{n \times n}$. For $t \geq 0$, a $(K,t)$
sparsifier can be constructed for $\lat$ using $2^{O(n)} \poly(\cdot)$ time and $2^n \poly(\cdot)$ space.
\end{theorem}

The proof of the above theorem is the subject of Sections \ref{sec:lat-spar-rand} and \ref{sec:lat-spar-det} (randomized and deterministic
constructions, respectively). Using the above lattice sparsifier construction, we present the following simple algorithm for $(1+\eps)$-CVP.

\begin{algorithm}
\caption{Approx-Closest-Vectors($K,\lat,\vecx,\eps$)}
\label{alg:apx-cvp}
\begin{algorithmic}[1]
\REQUIRE $(\veczero,r,R)$-centered convex body $K \subseteq \R^n$ with weak distance oracle $D_K$ for $\length{\cdot}_K$, a basis $B \in \Q^{n \times n}$
for $\lat$, target $\vecx \in \Q^n$, $0 < \eps \leq 1$
\ENSURE Outputs a non-empty set $S \subseteq \set{\vecy \in \lat: \|\vecy-\vecx\|_K \leq (1+\eps) d_K(\lat, \vecx)}$
\STATE {\bf if} $\vecx \in \lat$ {\bf then return} \set{\vecx}
\STATE Compute $\vecz \in \CVP(B_2^n,\lat,\vecx)$ using the MV algorithm
\STATE $l \leftarrow \frac{\|\vecz-\vecx\|_2}{R}$; $\eps_0 \leftarrow \frac{\eps}{9} \min \set{1, l}$
\STATE $d \leftarrow \frac{l}{2}$; $\tilde{d}_x \leftarrow \infty$
\REPEAT 
\STATE $d \leftarrow 2d$
\STATE $\lat' \leftarrow $ Lattice-Sparsifier($K, \lat, \frac{\eps}{3}d$)
\FORALL{$\vecy \in$ Lattice-Enum$((1+\frac{\eps}{3})dK+\vecx,\lat',r\eps_0)$} 
	\STATE $\tilde{d}_x \leftarrow \min \set{\tilde{d}_x, D_{K,\eps_0} (\vecy-\vecx), (1+\frac{\eps}{3})d + \eps_0}$
\ENDFOR
\UNTIL{$\tilde{d}_x < \infty$}
\RETURN Lattice-Enum$((\tilde{d}_x+\eps_0)K+\vecx, \lat', r\eps_0)$
\end{algorithmic}
\end{algorithm}

\begin{theorem} \label{thm:approx-cvp} 
Algorithm \ref{alg:apx-cvp} (Approx-Closest-Vectors) is correct, and on inputs $K,\lat,\vecx,\eps$ (as above),
$K$ $\gamma$-symmetric, it runs in deterministic $2^{O(n)}\gamma^{-n}(1+\frac{1}{\eps})^n \poly(\cdot)$ time and $2^{n}\poly(\cdot)$ space.
\end{theorem}

\begin{proof}\hspace{1em}

\paragraph{Correctness:} If $\vecx \in \lat$, we are clearly done. Next since $K$ is $(0,r,R)$-centered, we have that
$\frac{\length{\vecy}}{R} \leq \length{\vecy}_K \leq \frac{\length{\vecy}}{r}$ for all $\vecy \in \R^n$. Now take any $\vecz \in \CVP(K,
\lat, \vecx)$ and $\tilde{\vecz} \in \SVP(B_2^n, \lat)$. Here we note that $d_x = \length{\vecz-\vecx}_K$. As in the algorithm, let $l =
\frac{\length{\tilde{\vecz}-\vecx}}{R}$. Now we see that
\[
l = \frac{\length{\tilde{\vecz}-\vecx}}{R} \leq \frac{\length{\vecz-\vecx}}{R} \leq \length{\vecz-\vecx}_K \leq
\length{\tilde{\vecz}-\vecx}_K \leq \frac{\length{\tilde{\vecz}-\vecx}}{r} = l \frac{R}{r}
\]
Therefore $l \leq d_x \leq l \frac{R}{r}$.

Let $d_f$ denote the value of $d$ after the first while loop terminates. We claim that $\frac{1}{2}d_f \leq d_x \leq (1+\eps/3)d_f +
\eps_0$. When the while loop terminates, we are guaranteed that the call to Lattice-Enum($(1+\frac{\eps}{3})d_fK+\vecx,\lat',r\eps_0$),
outputs a lattice vector in $\lat'$ at distance at most $(1+\frac{\eps}{3})d_f+\eps_0$ from $\vecx$. Since $\lat' \subseteq \lat$, we clearly
have that $d_x \leq (1+\frac{\eps}{3})d_f+\eps_0$ as needed. 

If the while loop terminates after the first iteration, then $d_f = l \leq d_x$ and hence $\frac{1}{2}d_f < d_x$ as needed. If the
loop iterates more than once, then for the sake of contradiction, assume that $\frac{1}{2} d_f > d_x$. Then in the last iteration, the value of 
$d$ is greater than $d_x$. Now we are guaranteed that Lattice-Sparsifier($K, \lat, \frac{\eps}{3}d$) returns a
lattice $\lat'$ satisfying 
\[
d_K(\lat', \vecx) \leq d_K(\lat, \vecx) + \frac{\eps}{3}d \leq (1+\frac{\eps}{3})d
\]
But then the call to Lattice-Enum($(1+\frac{\eps}{3})dK+\vecx,\lat',r\eps_0$) is guaranteed to return a lattice point, and hence the while
loop terminates at this iteration, a clear contradiction. Hence $\frac{1}{2}d_f \leq d_x$ as needed.

Let $d_x' = d_K(\lat',\vecx)$, for $\lat'$ at the end of the while loop. We now claim that $\tilde{d}_x$ (as in the algorithm) satisfies
$d_x' - \eps_0 \leq \tilde{d}_x \leq d_x' + \eps_0$. We first note that $\tilde{d}_x = 
\min \set{d_f + \eps_0, D_{K,\eps_0}(\vecz-\vecx)}$
from some $\vecz \in \lat'$. By the guarantees on $D_{K,\cdot}$, we get that 
\[
\tilde{d}_x = \min \set{d_f + \eps_0, D_{K,\eps_0}(\vecz-\vecx)} \geq \min \set{d_x', \|\vecz-\vecx\|_K - \eps_0} \geq d_x' - \eps_0 \text{,} 
\]
as needed. For the second inequality, we examine two cases. First assume that Lattice-Enum($d_fK+\vecx, \lat', r\eps_0$) outputs $\vecz \in
\CVP(K,\lat',\vecx)$.  Then $\tilde{d}_x \leq D_{K,\eps_0}(\vecz-\vecx) \leq d_x' + \eps_0$ as needed. If Lattice-Enum does not output any
element of $\CVP(K,\lat,\vecx)$, we must have that $d_f < d_x'$ and hence $\tilde{d}_x \leq d_f + \eps_0 < d_x' + \eps_0$, as needed.  Finally
by the construction of $\lat'$, we also have that $d_x' \leq d_x + \eps/3 d_f \leq (1+ 2\eps/3)d_x$.

Since $d_x' \leq \tilde{d}_x+\eps_0$, we know that $((\tilde{d}_x+\eps_0)K + \vecx) \cap \lat \neq \emptyset$. Therefore we are guaranteed
that the final call to Lattice-Enum($(\tilde{d}_x+\eps_0)K+\vecx,\lat',r\eps_0$) outputs all the closest vectors of $\lat'$ to $\vecx$.
Finally, any vector $\vecy$ outputted during this call satisfies 
\[
\|\vecy-\vecx\|_K \leq \tilde{d}_x + 2\eps_0 \leq d_x' + 3\eps_0 \leq (1+2\eps/3)d_x + (\eps/3)l \leq (1+\eps)d_x
\]
as needed.

\paragraph{Running Time:} We first bound the running time of each call to Lattice-Enum. Within the while loop, the calls to
Lattice-Enum($(1+\eps/3)dK + \vecx,\lat',r\eps_0$) run in $2^{O(n)} G((1+\eps/3)dK,\lat') \poly(\cdot)$ time and $2^{n}\poly(\cdot)$ space.
By Lemma \ref{lem:gkl-bnds}, since $(1+\eps/3) = t(\eps/3)$ for $t = (3/\eps+1)$, we have that
\[
G((1+\eps/3)dK,\lat') \leq (4t+2)^n G((\eps/3)d,\lat') = 6^n(1+2/\eps)^n G((\eps/3)d,\lat') = 2^{O(n)}\gamma^{-n}(1+1/\eps)^n
\]
since by the guarantees on Lattice-Sparsifier, we have that $G((\eps/3)d,\lat') = \gamma^{-n} 2^{O(n)}$. Next the final call to
Lattice-Enum($(\tilde{d}_x+\eps_0)K+\vecx, \lat', r\eps_0$) runs $2^{O(n)} G((\tilde{d}_x+\eps_0)K,\lat') \poly(\cdot)$ time and
$2^{n}\poly(\cdot)$ space. Now note that $\eps_0 \leq \frac{1}{9} \eps d_x$, and hence $(1+\eps/3)d_f \geq d_x-\eps_0 \geq (1-\eps/9)d_x$. 
From here we get that
\[
d_f \geq \frac{1-\eps/9}{1+\eps/3} d_x \geq \frac{1-1/9}{1+1/3} d_x = 2/3 d_x
\] 
Finally, $\tilde{d}_x+\eps_0 \leq (1+\eps/3)d_f + 2\eps_0  \leq (1+\eps/3)d_f + 2/9 \eps d_x \leq (1+2\eps/3)d_f$. Therefore,
since $(1+2\eps/3) = t(\eps/3)$ for $t = (2+3/\eps)$, we get that
\begin{align*}
G((\tilde{d}_x+\eps_0)d_fK,\lat') &\leq G((1+2\eps/3)d_fK,\lat') \leq (4t+2)^n G((\eps/3)d_f,\lat') \\
                                  &= (10+12/\eps)^n G((\eps/3)d_f,\lat') = 2^{O(n)}\gamma^{-n}(1+1/\eps)^n
\end{align*}
by the guarantee on $\lat'$.

Lastly, note that each call to Lattice-Sparsifier takes at most $2^{O(n)}\poly(\cdot)$ time and $2^n \poly(\cdot)$ space. Since the while loop iterates
polynomially many times (i.e.~at most $\log_2(2R/r)$),the total runtime is $2^{O(n)}\gamma^{-n}(1+1/\eps)^n\poly(\cdot)$ and the total space
usage is $2^{n}\poly(\cdot)$ as needed.
\end{proof}

\section{A Simple Randomized Lattice Sparsifier Construction}
\label{sec:lat-spar-rand}

We begin with an existence proof for lattice sparsifiers using the probabilistic method. We will use the Cauchy-Davenport  sumset inequality and another lemma in number theory about primegaps, a consequence of a theorem of Rosser and Schoenfeld~\cite{journal/ijm/RS62,Nark00}.\footnote{The authors are indebted to J\'anos Pintz for finding these references.}

\begin{theorem}\label{thm:sumset} Let $p \geq 1$ be a prime. Then for $A_1,\dots,A_k \subseteq \Z_p$, we have that
\[
|A_1+\dots+A_k| \geq \min \set{p, \sum_{i=1}^k |A_i| - k + 1}
\]
\end{theorem}


\begin{lemma}\label{lem:primegap} For $x>1000$ 
there exists a prime $p \in \Z$ satisfying $x < p < \frac{4x}{3}$. 
\end{lemma}
\begin{proof}[Proof of Lemma \ref{lem:primegap} (Prime Gap)]
We will use the bounds $\pi(x)>x/\ln(x)$ if $x>17$, and $\pi(x)<1.25506x/\ln(x)$ if $x>1$
where $\pi(x)$ denotes the number of primes $<x$~\cite{journal/ijm/RS62,Nark00}. If $x>1000$ then $\pi(4x/3)> (4x/3)/\ln(4x/3)>1.25506x/\ln(x)>\pi(x)$, the lemma follows.
\end{proof}

%
%
%
%
%

We begin with the following crucial lemma. This forms the core of our lattice sparsifier construction.

\begin{lemma}
\label{lem:good-vec}
Let $p$ be a prime and $S \subseteq \Z_p^n$ satisfying $1000 < |S| < p < \frac{4|S|}{3}$ 
and $\veczero \in S$. Then there exists $\veca \in \Z_p^n$ satisfying
\begin{enumerate}
\item $|\set{\vecy \in S: \pr{\vecy}{\veca} \equiv 0 \imod{p}}| \leq 6$
\item $|\set{\pr{\vecy}{\veca} \imod{p}: \vecy \in S}| \geq \frac{p+2}{3}$
\end{enumerate}
\end{lemma}

\begin{proof}
Let $\veca$ denote a uniform random vector in $\Z_p^n$.  We will show that $\veca$ satisfies both conditions (1)
and (2) with non-zero probability. Let $E_i^{\vecy}$ denote the indicator of the event $\pr{\veca}{\vecy} \equiv i$ for $\vecy \in S$ and $i \in \Z_p$.

\paragraph{Claim 1: $\E[\sum_{\vecy \in S \minuszero} E_0^{\vecy}] = \frac{|S|-1}{p}$}
\begin{proof}
By linearity of expectation it suffices to prove $\E[E_0^{\vecy}] = \Pr[\pr{\veca}{\vecy}] = \frac{1}{p}$ for $\vecy \in S \minuszero$.
Since $\vecy \neq \veczero$, $p$ is a prime, and $\veca$ is uniform in $\Z_p^n$ we have that $\pr{\veca}{\vecy}$ is uniform in $\Z_p$.
Therefore $\Pr[\pr{\veca}{\vecy}] = \frac{1}{p}.$
\end{proof}

\paragraph{Claim 2: $\E[\sum_{\vecx, \vecy \in S, \vecx \neq \vecy} E_0^{\vecx-\vecy}] = \frac{|S|^2-|S|}{p}$}
\begin{proof}
If $\vecx \neq \vecy$ then $\E E_0^{\vecx-\vecy} = \frac{1}{p}$. The Claim follows by the linearity of expectation.
\end{proof}

Now we will choose the vector $\veca \in \Z_p^n$. By Markov's inequality

$\Pr[ |\set{ \vecy \in S \minuszero : \pr{\veca}{\vecy} \equiv 0 }| < 6 ] \geq 1-\frac{|S|-1}{6p} > \frac{5}{6}$, and
 
$\Pr[|\set{ (\vecx, \vecy):  \vecx, \vecy \in S, \vecx \neq \vecy, \pr{\veca}{\vecx} \equiv \pr{\veca}{\vecy}}| 
\leq \frac{6|S|}{5} ] \geq 1-\frac{5|S|^2-5|S|}{6|S|p} > \frac{1}{6}$. 

Hence there exists an $\veca$ such that both events hold. 
The first condition of the lemma is easy to check:

\noindent
$|\set{\vecy \in S: \pr{\vecy}{\veca} \equiv 0 }|=|\set{\vecy \in S \minuszero: \pr{\vecy}{\veca} \equiv 0 }|+1 \leq 5+1=6$. \\ Now we will prove the second condition using our assumption and the Cauchy-Schwartz inequality:

\noindent
$\frac{11|S|}{5} \geq |\set{ (\vecx, \vecy):  \vecx, \vecy \in S, \vecx \neq \vecy, \pr{\veca}{\vecx} \equiv \pr{\veca}{\vecy}}|+ |S|=   |\set{ (\vecx, \vecy):  \vecx, \vecy \in S, \pr{\veca}{\vecx} \equiv \pr{\veca}{\vecy}}| \\
=\sum_{z \in \Z_p} |\set{\vecy \in S: \pr{\veca}{\vecy} \equiv z  }|^2 
\geq |S|^2/|\set{\pr{\vecy}{\veca} \imod{p}: \vecy \in S}|$. 
These yield 

$|\set{\pr{\vecy}{\veca} \imod{p}: \vecy \in S}| > \frac{5|S|}{11} > \frac{15p}{44} > \frac{p+2}{3}$.
\end{proof}

We give now our first lattice sparsifier construction. While this theorem is stated for symmetric norms only, it can be easily extended to general norms (see Lemma~\ref{lem:spar-equiv}).

\begin{theorem} 
\label{thm:spar-exist}
Let $K \subseteq \R^n$ be a symmetric convex body, $\lat \subseteq \R^n$ an $n$-dimensional lattice, and $t \geq 0$ a non-negative number. Let 
$N = |tK \cap \lat|$, and take a prime $p$ satisfying $N < p < \frac{4N}{3}$ if $N > 1000$ and $p = 3$ otherwise. Then there exists $\vecw \in \lat^*$ such that the 
sublattice $\lat(\vecw) = \set{\vecy \in \lat: \pr{\vecw}{\vecy} \equiv 0 \imod{p}}$ satisfies 
\begin{enumerate}
\item $\forall \vecx \in \R^n$, $d_K(\lat(\vecw),\vecx) \leq d_K(\lat(\vecw),\vecx) + 3t$
\item $G(3tK,\lat(\vecw)) \leq 1000 \cdot 7^n$
\end{enumerate}
\end{theorem}
\begin{proof}
If $N \leq 1000$, let $\vecw = \veczero$, so $\lat(\veczero) = \lat$. Condition (2) is trivially satisfied, and
for condition (1) Lemma \ref{lem:gkl-bnds} implies

\[
G(3tK,\lat) \leq (2 \cdot 3+1)^n|tK \cap \lat| \leq 1000 \cdot 7^n.
\]

Now we assume that $N > 1000$. By Lemma~\ref{lem:primegap} there exists a prime $p$
satisfying $N < p < \frac{4N}{3}$, as required by the theorem. Let $B^*=(\vecb^1,\dots,\vecb^n)$ denote a basis for $\lat^*$.  Set $S = \set{B^{*T}\vecy \imod{p\Z^n}: \vecy \in tK \cap \lat}$.

\paragraph{Claim: $|\lat \cap tK| = |S|$.} 
\begin{proof}
Clearly $|S| \leq |\lat \cap tK|$.
We will prove  $|S| \geq |\lat \cap tK|$ by contradiction:
assume not and take $\vecy_1, \vecy_2 \in \lat \cap tK$, where $\vecy_1-\vecy_2 \in p\lat$. Set $\vecy=\vecy_1-\vecy_2$, so $\vecy \in 2tK$.   
Note that $(k/p)\vecy \in \lat$ for $k \in \Z$ and 
\[
\|(k/p)\vecy\|_{K} = \left|k/p\right| \|\vecy\|_{K} \leq 2t\left|k/p\right|
\]
by the symmetry of $K$. Hence for $|k| \leq \floor{p/2}$ we get $\|(k/p) \vecy\|_{K} \leq \frac{1}{2} 2t = t,$ i.e. $(k/p) \vecy 
\in tK$.
But then there are at least $2\floor{p/2}+1 \geq p > N$ distinct lattice points in $\lat \cap tK$, a contradiction.
\end{proof}

Since $\veczero \in S$, and $|S| < p < \frac{4|S|}{3}$, by Lemma \ref{lem:good-vec} there exists $\veca \in \Z_p^n$ s.t. $|{\vecy \in S: \pr{\veca}{\vecy} \equiv 0 \imod{p}}| \leq 6$ and $|{\pr{\veca}{\vecy} \imod{p}: \vecy \in S}| 
\geq \frac{p+2}{3}$. Let $\bar{\veca}$ denote the unique representative of $\veca$ in $\set{0,\dots,p-1}^n$, 
and let $\vecw = B^* \bar{\veca}$. 

Let $S_{in} = \set{\vecy \in S: \pr{\veca}{\vecy} \equiv 0 \imod{p}}$ and 
$C = \set{\pr{\veca}{\vecy} \imod{p}:
\vecy \in S}$. We know that $|S_{in}| \leq 6$ and $|C| \geq \frac{p+2}{3}$ by our guarantees on $\veca$.
We establish condition (2) first. 
We know that $|tK \cap \lat(\vecw)| = |S_{in}| \leq 6$. Lemma
\ref{lem:gkl-bnds} implies
\[
G(3tK, \lat(\vecw)) \leq 7^n \cdot |tK \cap \lat(\vecw)| \leq 7^n \cdot 6  \leq 1000 \cdot 7^n \text{.}
\]

Now we establish condition (1), i.e.~for any $\vecx \in \R^n$, $d_K(\lat(\vecw), \vecx) \leq d_K(\lat, \vecx) + 3t$. 
Let $\vecy \in \lat$ be (one of) the closest vector(s) to $\vecx$, i.e.  $d_K(\lat, \vecx)= \|\vecy-\vecx\|_K$.  
Since $C \subseteq \Z_p, |C| \geq \frac{p+2}{3}$ Theorem \ref{thm:sumset} yields
\[
|C+C+C| \geq \min \set{p, 3(\frac{p+2}{3}+1)-3} \geq p,
\]
and hence $C+C+C = \Z_p$. Therefore, there exists $\vecy_1, \vecy_2, \vecy_3 \in tK \cap \lat$ and
$\vecz \in \lat(\vecw)$ satisfying $\vecy=\vecz+\vecy_1+\vecy_2+\vecy_3$.
Finally, by the triangle inequality and the symmetry of $K$ we get that \newline
$\|\vecz-\vecx\|_K \leq \|\vecy-\vecx\|_K + \|\vecz-\vecy\|_K \leq d_K(\lat, \vecx) + \sum_{i=1}^3 \|-\vecy_i\|_K \leq d_K(\lat, \vecx) + 3t$, as needed.
\end{proof}

\section{Derandomizing the Lattice Sparsifier Construction}
\label{sec:lat-spar-det}

We begin with a high level outline of the deterministic sparsifier construction. To recap, in the previous section, we build a $(K,t)$ sparsifier for $\lat$ as follows
\begin{enumerate}[itemsep=0pt]
\item Compute $N \leftarrow |tK \cap \lat|$. 
If $N \leq 1000$ then return $\lat'=\lat$.
Else find a prime $p$ satisfying $N < p < \frac{4N}{3}$.
\item Build basis $B^* \in \Q^{n \times n}$ for $\lat^*$ and compute 
$S \leftarrow \set{B^{*T} \vecy \imod{p}: \vecy \in tK \cap \lat}$.
\item Find a vector $\veca \in \Z_p^n$ satisfying (in fact, for slightly worse parameters, a random $\veca \in \Z_p^n$ succeeds with constant probability) 
\vspace{-10pt}
\[
(a) ~~ |\set{\vecy \in S: \pr{\veca}{\vecy} \equiv 0 \imod{p}}| \leq 6 \quad \quad (b) ~~ |\set{\pr{\veca}{\vecy}: \vecy \in S}| \geq \frac{p+2}{3}
\vspace{-10pt}
\]
\item Return sublattice $\lat' = \set{\vecy \in \lat: \pr{\vecy}{B^*\veca} \equiv 0 \imod{p}}$.
\end{enumerate}

To implement the above construction efficiently and deterministically, we must overcome several obstacles. First, the number of lattice points $N$ in
$tK \cap \lat$ could be very large (since we have no control on $t$). Hence we can not hope to compute $N$ or the set $S$ efficiently via lattice
point enumeration. Second, the construction of the vector $\veca$ is probabilistic (see Lemma \ref{lem:good-vec}): we must replace this with an
explicit deterministic construction.

To overcome the first difficulty, we will build the $(K,t)$ sparsifier iteratively. In particular, we will compute a sequence of sparsifiers
$\lat'_1,\dots,\lat'_k$, satisfying that $\lat'_{i+1}$ is a $(K, c^i \lambda)$ sparsifier for $\lat'_i$ for $i \geq 0$, where $\lat'_0 = \lat$,
$\lambda = \lambda_1(K,\lat)$ and $c > 1$ is a constant. We start the sparsification process at the minimum distance of $\lat$. We only increase the
sparsification distance by a constant factor at each step. Hence we will be able to guarantee that the number of lattice points we process at each
step is $2^{O(n)}$. Furthermore, the geometric growth rate in the sparsification distance will allow us to conclude that $\lat'_i$ is in fact a $(K,
\frac{c^{i+1}}{c-1}\lambda)$ sparsifier for $\lat$. Hence, iterating the process roughly $k \approx \ln \frac{t}{\lambda_1}$ times will yield the
final desired sparsifier.

For the second difficulty, i.e.~the deterministic construction of $\veca$, the main idea is to use a dimension reduction procedure which allows
$\veca$ to be computed efficiently via exhaustive enumeration (i.e.~trying all possible $\veca$'s).  Let $N$ and $S$ be as in the description. Since
$N < p < \frac{4N}{3}$, we note that an exhaustive search over $\Z_p^n$ requires a search over $p^n \leq (\frac{4N}{3})^n$ possibilities, and the
validity check (i.e.~conditions $(a)$ and $(b)$) for any particular $\veca$ can be implemented in $\poly(N)$ time by simple counting. Since the
existence of the desired $\veca$ depends only on $|S|$ and $p$ (and not on $n$), if we can compute a linear projection $\pi:\Z_p^n \rightarrow
\Z_p^{n-1}$ such that $\pi(S) = |S|$, then we can reduce the problem to finding a good $\veca \in \Z_p^{n-1}$ for $\pi(S)$. Indeed, such a map $\pi$
can be computed efficiently and deterministically as long as $n \geq 3$. To see this, we first identify full rank $n-1$ dimensional projections with
their kernels, i.e.~lines in $\Z_p^n$.  From here, we note that distinct elements $\vecx,\vecy \in S$ collide under the projection induced by a line
$l$ iff $\vecx-\vecy \in l$. Since the total number of lines spanned by differences of elements in $S$ is at most $\binom{|S|}{2} < \binom{p}{2}$, as
long as there are at least $\binom{p}{2}$ lines in $\Z_p^n$ (i.e.~for $n \geq 3$) we can compute the desired projection. 
Therefore, repeating the process $n-2$ times, we are left with finding a good $\veca \in \Z_p^2$, which we can do by trying all $p+1 < \frac{4N}{3}+1$
lines in $\Z_p^2$. As discussed in the previous paragraph, we will be able to guarantee that $N = 2^{O(n)}$, and hence the entire construction
described above can be implemented in $2^{O(n)}$ time and space as desired.

\subsection{Algorithms}

We begin with the deterministic algorithm implementing Lemma \ref{lem:good-vec}. We denote the set of lines in $\Z_p^n$ by ${\rm Lines}(\Z_p^n)$.
For a vector $\vecq \in
\Z_p^n$ we denote its orthogonal complement by $\vecq^{\perp} = \set{\vecy \in \Z_p^n: \pr{\vecq}{\vecy} \equiv 0 \imod{p}}$.

\begin{algorithm}[h]
\caption{Algorithm Good-Vector($S$, $p$)}
\label{alg:good-vector}
\begin{algorithmic}[1]
\REQUIRE $S \subseteq \Z_p^n$, $\veczero \in S$, integer $n \geq 1$, $p$ a prime satisfying 
$1000<|S| < p < \frac{4|S|}{3}$.
\ENSURE $\veca \in \Z_p^n$ satisfying conditions of Lemma \ref{lem:good-vec}~.
\STATE {\bf if} $n = 1$, {\bf return} $1$
\STATE $P \leftarrow I_n$ ($n \times n$ identity)
\FOR{$n_0$ {\bf in} $n$ {\bf to} $3$}
	\FORALL{$\vecq \in {\rm Lines}(\Z_p^{n_0})$}
		\STATE Compute basis $B \in \Z_p^{n_0 \times n_0-1}$ satisfying $\vecq^{\perp} = B \Z_p^{n_0-1}$
    \STATE $\forall$ distinct $\vecx, \vecy \in PS$ check that $B^T \vecx \not\equiv B^T \vecy \imod{p\Z^{n_0-1}}$.  \newline
           If no collisions, set $P \leftarrow B^TP$ and exit loop; otherwise, continue.
	\ENDFOR
\ENDFOR
\FORALL{$\vecq \in {\rm Lines}(\Z_p^2)$}
\STATE Pick $\veca \in \vecq \minuszero$
	\STATE Compute $zeros \leftarrow |\set{\vecy \in PS: \pr{\veca}{\vecy} \equiv 0 \imod{p}}|$
  \STATE Compute $distinct \leftarrow |\set{\pr{\veca}{\vecy} \imod{p}: \vecy \in PS}|$
	\IF{$zeros \leq 6$ and $distinct \geq \frac{p+2}{3}$}
		\RETURN $P^t \veca$
	\ENDIF
\ENDFOR
\end{algorithmic}
\end{algorithm}

For the desired application of the algorithm given below, the set $S$ above will in fact be represented implicitly. Here the main access methodology we will
require from $S$ is a way to iterate over its elements. In the context of $(1+\eps)$-CVP, the enumeration method over $S$ will correspond to the
Lattice-Enum algorithm. Here we state the guarantees of the algorithm abstractly in terms of the number of iterations required over $S$.

\begin{theorem}
\label{thm:good-vector}
Algorithm \ref{alg:good-vector} is correct, and performs $\poly(n, \log p)p^4$ arithmetic operations and $O(n
p^3)$ iterations over the elements of $S$. Furthermore, the space usage (not counting the space needed to iterate over $S$) is $\poly(n, \log p)$.
\end{theorem}

\begin{proof}[Analysis of Good-Vector]\ \vspace{-12pt}
\paragraph{Correctness:} We must show that the outputted vector $\veca$ satisfies the guarantees of Lemma
\ref{lem:good-vec}: 
\begin{enumerate}
\item $|\set{\vecy \in S: \pr{\veca}{\vecy} \equiv 0 \imod{p}}| \leq 6$
\item $|\set{\pr{\veca}{\vecy} \imod{p}: \vecy \in S}| \geq \frac{p+2}{3}$
\end{enumerate}

If $n = 1$ then setting $\veca \in \Z_p$ to $1$ (i.e.~line 1) trivially satisfies $(1)$ and $(2)$. We assume $n \geq 2$. We
prove the following invariant for the first loop (line 2): at the beginning of each iteration, $P \in \Z_p^{n_0 \times n}$ and $|PS| = |S|$. 

First let us assume that during the loop iteration, we find  $B \in \Z_p^{n_0 \times (n_0-1)}$ satisfying $B^T\vecx
\neq B^T\vecy$ for all distinct $\vecx, \vecy \in PS$ (verified in line 5). This yields that the map
$\vecx \rightarrow B^T\vecx$ is injective when restricted to $PS$, and hence $|B^TPS| = |S|$. Next, since $B \in \Z_p^{n_0 \times (n_0-1)}$ 
and $P \in \Z_p^{n_0 \times n}$, we have that $P$ is set to $B^TP \in \Z_p^{(n_0-1) \times n}$ for the next iteration, 
as needed.

Now we show that a valid projection matrix $B^T$ is guaranteed to exist as long as $n_0 \geq 3$. First, we claim that there
exists $\vecq \in {\rm Lines}(\Z_p^{n_0})$, such that for all distinct $\vecx,\vecy \in PS$, $(\vecq+\vecx) \cap
(\vecq + \vecy) = \emptyset$, i.e.~all the lines passing through $PS$ in the direction $\vecq$ are disjoint.
A line $\vecq$ fails to satisfy (a) if and only if $\vecq = \Z_p (\vecx-\vecy)$
for distinct $\vecx,\vecy \in PS$.  The number of lines that can be generated in this way 
from $PS$ is at most $\binom{|PS|}{2} = \binom{|S|}{2} < \frac{p(p-1)}{2}$.
Since $|{\rm Lines}(\Z_p^{n_0})| = \frac{p^{n_0}-1}{p-1} > \frac{p(p-1)}{2}$ for $n_0 \geq 3$
we may pick $\vecq \in {\rm Lines}(\Z_p^n)$ that satisfies (a). 
Now let $B \in \Z_p^{n_0 \times (n_0-1)}$ denote a basis satisfying $\vecq^{\perp} = B\Z_p^{n_0-1}$. 
We claim that $|B^TPS| = |PS|$.
Assume not, then there exists distinct $\vecx, \vecy \in PS$ such that 
\[
B^T\vecx \equiv B^T\vecy ~\Leftrightarrow~ B^T(\vecx-\vecy) \equiv \veczero ~\Leftrightarrow~ (\vecx-\vecy) \in (B\Z_p^{n_0-1})^\perp = \vecq \text{,}
\]
which contradicts our assumption on $\vecq$. Therefore, the algorithm is indeed guaranteed to find a valid projection, as needed.

After the first for loop, we have constructed $P \in \Z_p^{2 \times n}$ satisfying $|PS| = |S|$, where $|S| < p <
\frac{4|S|}{3}$.  By Lemma \ref{lem:good-vec}, there exists $\veca \in \Z_p^2$ satisfying $(1)$ and $(2)$ for the set $PS$. 
Since $(1)$ and $(2)$ holds for any non-zero multiple of $\veca$, i.e.~any vector defining the same line as $\veca$, we may restrict the search to elements of ${\rm Lines}(\Z_p^2)$. 
Therefore, by trying all $p+1$
elements of ${\rm Lines}(\Z_p^2)$ the algorithm is guaranteed to find a valid $\veca$ for the $PS$. Noting that
$\pr{\veca}{P\vecy} \equiv \pr{P^T \veca}{\vecy}$, we get that $P^T \veca$ satisfies $(1)$ and $(2)$ for the
set $S$, as needed.

\paragraph{Runtime:}
For $n = 1$ the runtime is constant. We assume $n \geq 2$. Here the first for loop is executed $n-2$ times. For each
loop iteration we run though $\vecq \in {\rm Lines}(\Z_p^{n_0})$ until we find one inducing a good projection matrix
$B$. From the above analysis, we iterate through at most $\binom{|S|}{2} < \frac{p(p-1)}{2}$ elements $\vecq \in {\rm
Lines}(\Z_p^{n_0})$ before finding a good projection matrix. For each $\vecq$, we build a basis matrix $B$ for
$\vecq^{\perp}$ which can be done using $\poly(n, \log p)$ arithmetic operations. Next, we check
for collisions against each pair $\vecx,\vecy \in PS$, which can be done using $O(|S|) = O(p)$ iterations over $S$.
Therefore, at each loop iteration we enumerate over $S$ at most $p^3$ times while performing only polynomial time
computations. Hence, the total number of operations (excluding the time needed to output the elements of $S$) is at most
$\poly(n, \log p) p^4$.

For the last phase, we run through the elements in ${\rm Lines}(\Z_p^2)$, where $|{\rm Lines}(\Z_p^2)| = p+1$. The
validity check for $\veca \in {\rm Lines}(\Z_p^2)$ requires computing both the quantities $(1)$ and $(2)$. To
compute $|\set{\vecy \in S: \pr{\vecy}{\veca} \equiv 0 \imod{p}}|$ we iterate once over the set $S$ and count how many
zero dot products there are. To compute $|\set{\pr{\veca}{\vecy}: \vecy \in S}|$, we first iterate over all residues in
$\Z_p$. Next, for each residue $i \in Z_p$, if we find $\vecy \in S$ satisfying $\pr{\veca}{\vecy} \equiv i \imod{p}$, we
increment our counter by one, and otherwise continue. Hence for any specific $\veca \in \Z_p^2$, we iterate over the set $S$
exactly $p+1$ times, performing $\poly(n, \log p) p^2$ operations. Hence, over the whole loop we perform $O(p^2)$ iterations over
the set $S$, and perform $\poly(n, \log p) p^3$ operations. 

Therefore, over the whole algorithm we iterate over the set $S$ at most $n p^3$ times, and perform at most $\poly(n,
\log p)p^4$ operations. Furthermore, not counting the space needed to iterate over the set $S$, the space used by
the algorithm is $\poly(n, \log p)$.
\end{proof}

Before moving into the derandomized sparsifier construction, we show a simple equivalence between building a sparsifier for symmetric and asymmetric norms.

\begin{lemma}
\label{lem:spar-equiv}
Let $K$ be a $\gamma$-symmetric convex body, and let $\lat$ be an n-dimensional lattice. Take $\lat' \subseteq \lat$, a full dimensional sublattice.
Then for $t \geq 0$, we have that $\lat'$ is a $(K \cap -K,t)$ sparsifier $\Rightarrow \lat'$ is a $(K,t)$ sparsifier.
\end{lemma}
\begin{proof}  
Let $\lat' \subseteq \lat$ be a $(K \cap -K,t)$ sparsifier. Since $K \cap -K$ is $1$-symmetric,
by definition we have that $G(t(K \cap -K), \lat') = 2^{O(n)}$. By Lemma \ref{lem:cov-est} and $\gamma$-symmetry of $K$, we have that 
\begin{align*}
N(tK, t(K \cap -K) &= N(K, K \cap -K) \leq \frac{\vol_n(K + \frac{1}{2} (K \cap -K))}{\vol_n(\frac{1}{2}(K \cap -K)} 
                   \leq  \frac{\vol_n(\frac{3}{2} K)}{\vol_n(\frac{1}{2}(K \cap -K))} \leq 3^n \gamma^{-n}
\end{align*}
Therefore
\[
G(tK, \lat') \leq   G(t(K \cap -K), \lat') N(tK, t(K \cap -K)) = 2^{O(n)} 3^n \gamma^{-n} = 2^{O(n)} \gamma^{-n} \text{ as needed.}
\]
Since $K \cap -K \subseteq K$, we note that $\|\veca\|_K \leq \|\veca\|_{K \cap -K}$ for all $\veca \in \R^n$. Now take $\vecx \in \R^n$, and
take $\vecz \in \CVP(K,\lat,\vecx)$. By the guarantee on $\lat'$, there exists $\vecy \in \lat'$ such that 
\[
\|\vecy-\vecz\|_{K \cap -K} \leq d_{K \cap -K}(\lat,\vecz) + t = t
\]
since $\vecz \in \lat$. Next, using the triangle inequality we have that
\[
\|\vecy-\vecx\|_K \leq \|\vecy-\vecz\|_K + \|\vecz-\vecx\|_K \leq \|\vecy-\vecz\|_{K \cap -K} + d_K(\lat,\vecx) \leq d_K(\lat,\vecx) + t
\]
as needed. Therefore, $\lat'$ is a $(K,t)$ sparsifier for $\lat$ as claimed.
\end{proof}

From the above lemma, we see that it suffices to build lattice sparsifiers for symmetric convex bodies, i.e. to build a $(K,t)$ sparsifier it suffices
to build a $(K \cap -K, t)$ sparsifier for $\lat$. 

We now show how to use the Good-Vector algorithm to get a completely deterministic Lattice Sparsifier construction. The correctness and runtime of the
algorithm given below yields the proof of Theorem~\ref{thm:lattice-sparsifier}.

\begin{algorithm}[h]
\caption{Algorithm Lattice-Sparsifier($K$, $\lat$, $t$)}
\label{alg:lattice-sparsifier}
\begin{algorithmic}[1]
\REQUIRE $(\veczero, r,R)$-centered convex body $K \subseteq \R^n$ with distance oracle $D_{K, \cdot}$ for $\|\cdot\|_K$, basis
$B \in \Q^{n \times n}$ for $\lat$, and $t \geq 0$.
\ENSURE $(K,t)$ sparsifier for $\lat$
\STATE $K \leftarrow K \cap -K$
\STATE Compute $\vecy \in \text{Shortest-Vectors}(K,\lat,\frac{1}{3})$
\STATE $\lambda \leftarrow D_{K,\frac{1}{2}}(\vecy)$; $\eps \leftarrow 7^{-(n+5)}$
\STATE $k \leftarrow \lfloor\ln \left(\frac{2}{3}\frac{t}{\lambda} + 1 \right) / \ln 3\rfloor$
\STATE $\lat_0 \leftarrow \lat; B_0 \leftarrow B$
\FOR{$i$ {\bf in} $0$ {\bf to} $k-1$}
	\STATE $S \leftarrow \text{Lattice-Enum}(3^i(1-\eps)\lambda K, \lat_i, \eps \lambda r)$ 
	\STATE Compute $N \leftarrow |S|$
	\IF{$N > 1000$}
		\STATE Compute $B^*_i \leftarrow B^{-T}_i$, a basis for $\lat^*_i$
		\STATE Compute prime $p$ satisfying $N < p < \frac{4N}{3}$
		\STATE $\veca \leftarrow \text{Good-Vector}(B^{*T}_i S \imod{p \Z^n}, p)$
		\STATE Compute $\lat_{i+1} \leftarrow \set{\vecy \in \lat_{i}: \pr{B^*_i \veca}{\vecy} \equiv 0 \imod{p}}$ and basis $B_{i+1}$ for $\lat_{i+1}$
	\ELSE
		\STATE $\lat_{i+1} \leftarrow \lat_i; B_{i+1} \leftarrow B_i$
	\ENDIF
\ENDFOR 
\RETURN $\lat_k$
\end{algorithmic}
\end{algorithm}

\begin{proof}[Proof of Theorem~\ref{thm:lattice-sparsifier} (Lattice Sparsifier Construction)]\ \vspace{-12pt}

\paragraph{Correctness:} We show that the outputted lattice is a $(K,t)$ sparsifier for $\lat$. By Lemma
\ref{lem:spar-equiv} it suffices to show that the algorithm outputs a $(K \cap -K, t)$ sparsifier, which justifies
the switch in line 2 from $K$ to $K \cap -K$. In what follows, we therefore assume that $K$ is symmetric.

We first claim that $\lambda \leq 2 \lambda_1(K,\lat)$. To see by the guarantee on Shortest-Vector($K,\lat,\frac{1}{3}$), we have that $\|y\|_K \leq \frac{4}{3}\lambda_1(K,\lat)$. This implies 
\[
\lambda \leq \frac{3}{2}\|\vecy\|_K \leq \frac{3}{2} \cdot \frac{4}{3} \lambda_1(K,\lat) = 2 \lambda_1(K,\lat) \text{,}
\]
as needed.

\paragraph{Claim 1:} for each $i$, $0 \leq i \leq k$, we have that
\begin{enumerate}
\item $\forall \vecx \in \R^n$, $d_K(\lat_i,\vecx) \leq d_K(\lat,\vecx) + \frac{3}{2}(3^i-1) \lambda$.
\item $G(3^i \lambda, \lat_i) \leq 7^{n+4}$.
\end{enumerate}
\begin{proof}
We establish the claim by induction on $i$. For $i = 0$, we have that $\lat_0 = \lat$. Therefore, $\lat_0$ trivially satisfies property $(1)$. Next, since $\lambda \leq 2\lambda_1(K,\lat)$, by Lemma \ref{lem:gkl-bnds} we have that 
$G(\lambda K, \lat_0) \leq (2 \cdot 2 + 1)^n = 5^n < 7^{n+4}$. Hence $\lat_0$ also satisfies $(2)$.

We now prove the claim for $i \geq 1$. Let $S$ denote the set outputted by 
Lattice-Enum($3^{i-1}(1-\eps)\lambda K, \lat_{i-1}, \eps \lambda r$). 
By the guarantees on Lattice-Enum, the set $S$ satisfies 
$\displaystyle 3^{i-1}(1-\eps)\lambda K \cap \lat_{i-1} \subseteq S \subseteq 
(3^{i-1}(1-\eps)\lambda K + \eps\lambda r B_2^n) \cap \lat_{i-1}$. Since $rB_2^n \subseteq K$ and $i \geq 1$ 
we have $3^{i-1}(1-\eps)\lambda K + \eps\lambda rB_2^n \subseteq 3^{i-1}\lambda K$. Therefore,
\begin{equation}
\label{eq:lsc-1}
3^{i-1}(1-\eps)\lambda K \cap \lat_{i-1} \subseteq S \subseteq 3^{i-1}\lambda K \cap \lat_{i-1}
\end{equation}

Set $N=|S|$ (line 8). By \eqref{eq:lsc-1} and the induction hypothesis we have 
\[
|3^{i-1}(1-\eps)\lambda K \cap \lat_{i-1}| \leq N \leq |3^{i-1}\lambda K \cap \lat_{i-1}| \leq G(3^{i-1}\lambda K, \lat) \leq 7^{n+4}
\]

Assume $N \leq 1000$. Then the algorithm sets $\lat_i = \lat_{i-1}$ and $B_i = B_{i-1}$. The induction hypothesis implies
for $\vecx \in \R^n$ that
\[
d_K(\lat_i, \vecx) = d_K(\lat_{i-1}, \vecx) \leq d_K(\lat, \vecx) + \frac{3}{2}(3^{i-1}-1)\lambda \leq d_K(\lat, \vecx) + \frac{3}{2}(3^i-1)\lambda \text{,}
\]
and hence $\lat_i$ satisfies $(1)$. Next, by \eqref{eq:lsc-1} we have that $|3^i(1-\eps)\lambda K \cap \lat_i| \leq N \leq 1000$.
Therefore, Lemma \ref{lem:gkl-bnds} yields
\begin{align*}
G(3^{i+1} \lambda K, \lat_{i+1}) \leq (2\cdot3(1/(1-\eps))+1)^n |3^i(1-\eps)\lambda K \cap \lat_{i+1}| \\
                             \leq 7^n(1+2\eps)^n \cdot 1000 \leq 7^{n+4} \text{,}
\end{align*}
where the last two inequalities follow since $\eps \leq 7^{-(n+5)}$. Therefore $\lat_i$ satisfies requirement $(2)$ as needed.

Assume $N > 1000$. Here we first compute $N < p < \frac{4N}{3}$, and a dual basis $B_{i-1}^*$ for $\lat_{i-1}^*$. 

\paragraph{Claim 2: $|B^{*T}_{i-1} S \imod{p \Z^n}| = N$}
\begin{proof}
Since $|S| = N$, if the claim is false, there exists distinct $\vecx,\vecy \in \lat$ such that 
\[
B^{*T}_{i-1}\vecx \equiv B^{*T}_{i-1}\vecy \imod{p\Z^n} \Leftrightarrow B^{*T}_{i-1}(\vecx-\vecy) \equiv \veczero \imod{p\Z^n}
\Leftrightarrow \vecx-\vecy \in p \lat_{i-1} \text{.}
\]
Since $\vecx,\vecy \in 3^{i-1} \lambda K$ and $K$ is symmetric, we have that $\vecx-\vecy \in 2\cdot3^{i-1}K \cap p \lat_{i-1}$. Let $\vecz =
\vecx-\vecy \in p \lat_{i-1}$. We examine the vector $s\frac{\vecz}{p}$ for $s \in \Z$ satisfying $|s| \leq \floor{\frac{p}{2}} = \frac{p-1}{2}$ 
(since $p$ is odd). Since $\frac{\vecz}{p} \in \lat_{i-1}$, we have that $s \frac{\vecz}{p} \in \lat_{i-1}$ and 
\begin{align*}
s \frac{\vecz}{p} \in \left|\frac{s}{p}\right|\cdot 2\cdot 3^{i-1} K &\subseteq \left(\frac{p-1}{2p}\right)2 \cdot 3^{i-1} K = \left(1-\frac{1}{p}\right)3^{i-1}K \\
                                                          &\subseteq (1-\eps)3^{i-1}K,
\end{align*}
where the last inequality follows since $p < \frac{4N}{3} \leq \frac{4}{3} \cdot 7^{n+4}$ and $\eps = 7^{-(n+5)}$. 
Then, since $s$ can take $2\floor{\frac{p}{2}}+1
= p$ different values, the set $(1-\eps)3^{i-1}K$ contains at least $p$ lattice points in $\lat_{i-1}$. However, by the construction of $N$, we have that \newline
$|(1-\eps)3^{i-1}K \cap \lat_{i-1}| \leq N < p$, a clear contradiction. The claim thus holds.
\end{proof}

Next, the algorithm computes $\veca \leftarrow $ Good-Vector($B^{*T}_iS \imod{p \Z^n},p$), and sets \newline $\lat_i = \set{\vecy \in \lat: \pr{B^* \veca}{\vecy}
\equiv 0 \imod{p}}$. From Claim 2, equation \ref{eq:lsc-1} and the guarantees on Good-Vector, we get
\begin{enumerate}
\item $|3^{i-1}(1-\eps)\lambda K \cap \lat_i| = |\set{\vecy \in 3^{i-1}(1-\eps)\lambda K \cap \lat_{i-1}: \pr{B^*\veca}{\vecy} \equiv 0 \imod{p}}| \leq 6$.
\item $|\set{\pr{B^*\veca}{\vecy} \imod{p}: \vecy \in 3^{i-1}\lambda K \cap \lat_{i-1}}| \geq \frac{p+2}{3}$.
\end{enumerate}
From here, using the identical analysis as in Theorem \ref{thm:spar-exist}, from $(a)$ above we get that $\forall \vecx \in \R^n$,
$d_K(\lat_i, \vecx) \leq d_K(\lat_{i-1},\vecx) + 3 \cdot 3^{i-1}\lambda$. The induction hypothesis on $\lat_{i-1}$ implies
\[
d_K(\lat_{i-1},\vecx) + 3^i \lambda \leq d_K(\lat, \vecx) + \frac{3}{2}(3^{i-1}-1) \lambda + 3^i\lambda = d_K(\lat,\vecx) + \frac{3}{2}(3^i-1)\lambda \text{.}
\]
Therefore $\lat_i$ satisfies $(1)$ as needed. Using $(b)$ and Lemma \ref{lem:gkl-bnds} we have that
\begin{align*}
G(3^i\lambda K, \lat_i) &\leq (2\cdot3\cdot(1/(1-\eps))+1)^n |3^{i-1}(1-\eps)\lambda K \cap \lat_i| \\
	                      &\leq 7^n(1+2\eps)^n \cdot 6 < 7^{n+4} \text{.}
\end{align*}
Therefore $\lat_i$ satisfies $(2)$. The claim thus follows.
\end{proof}

Given Claim 1, we will show that $\lat_k$ is a $(K,t)$ sparsifier for $\lat$. By our choice of $k$, note that $\frac{3}{2}(3^k-1)\lambda \leq t
\leq 3 \cdot \frac{3}{2}(3^{k+1}-1)\lambda$. By the claim, for $\vecx \in \R^n$, $d_K(\lat_k, \vecx) \leq d_K(\lat, \vecx) +
\frac{3}{2}(3^k-1)\lambda \leq d_K(\lat, \vecx) + t$. It therefore only remains to bound $G(tK, \lat_k)$. By the previous bounds
\[
\frac{t}{3^k\lambda} \leq \frac{3}{2} \frac{(3^{k+1} - 1)\lambda}{3^k\lambda} < \frac{9}{2}
\]
Therefore, the claim and Lemma \ref{lem:gkl-bnds} imply
\[
G(tK, \lat_k) \leq (2 \cdot \frac{9}{2} + 1)^n G(3^k \lambda K, \lat_k) \leq 10^n \cdot 7^{n+4} = 2^{O(n)}
\]
as needed. The algorithm returns a valid $(K,t)$ sparsifier for $\lat$.

\paragraph{Runtime:} The algorithm first runs the Shortest-Vectors on $K$ and $\lat$, which takes $2^{O(n)} \poly(\cdot)$ time and $2^n
\poly(\cdot)$ space. Next, the for loop on line $6$ iterates $k = \floor{\ln(\frac{2}{3}\frac{t}{\lambda}+1)/\ln 3} = \poly(\cdot)$ times.  

Each for loop iteration, indexed by $i$ satisfying $0 \leq i \leq k-1$, consists of computations over the set $S \leftarrow$
Lattice-Enum($3^i(1-\eps)\lambda K, \lat_i, \eps \lambda r)$. For the intended implementation, we do not store the set $S$ explicitly. Every time the algorithm needs to iterate over $S$, we implement this by performing a call to 
Lattice-Enum($3^i(1-\eps)\lambda K, \lat_i, \eps \lambda r)$. Furthermore, the algorithm only interacts with $S$ by iterating over its elements, and hence the implemented interface
suffices. Now at the loop iteration indexed by $i$, we do as follows:
\begin{enumerate}
\item Compute $N = |S|$. This is implemented by iterating over the elements of $S$ and counting, and so by the guarantees of Lattice-Enum
requires at most $2^{O(n)} G(3^i \lambda K, \lat_i) \poly(\cdot) = 2^{O(n)} \poly(\cdot)$ time (by Claim 1) and $2^{n} \poly(\cdot)$ space.
\item If $N \leq 1000$, we keep the same lattice and skip to the next loop iteration. If $N > 1000$, continue.
\item Compute $B^*_i = B^{-T}_i$. This can be done in $\poly(\cdot)$ time and space. 
\item Compute a prime $p$ satisfying $N < p < \frac{4N}{3}$. Such a prime can be computed by trying all integers in the previous
range and using trial division. This takes at most $O(N^2 \poly(\log N)) = 2^{O(n)}$ time and $\poly(n)$ space.
\item Call Good-Vector$(B^{T*} S \imod{p\Z^n}, p)$. By the guarantees on Good-Vector, the algorithm performs $\poly(n, \log p) p^4 = 2^{O(n)}$
operations and iterates at most $n p^3 = 2^{O(n)}$ times over the set $B^{T*} S \imod{p\Z^n}$. These iterations can be performed 
$2^{O(n)} \poly(\cdot)$ time and $2^n \poly(\cdot)$ space by the guarantees on Lattice-Enum.
\item Compute a basis $B_{i+1}$ for the new lattice $\lat_{i+1} = \set{\vecy \in \lat_i: \pr{B^{*T}\veca}{\vecy} \equiv 0 \imod{p}}$. This
can be done in $\poly(\cdot)$ time.
\end{enumerate}

From the above analysis, we see that the entire algorithm runs in $2^{O(n)} \poly(\cdot)$ time and $2^n \poly(\cdot)$ space as needed.
\end{proof} 

%
%

\section{Further Applications and Future Directions}
\label{sec:future-directions}
 
\paragraph{Integer Programming.} We explain how the techniques in this paper apply to Integer Programming (IP), i.e.~the problem of deciding whether a
polytope contains an integer point, and discuss some potential associated venues for improving the complexity of IP. For a brief history, the first
breakthrough works on IP are by Lenstra~\cite{lenstra83:_integ_progr_with_fixed_number_of_variab} and
Kannan~\cite{kannan87:_minkow_convex_body_theor_and_integ_progr}, where it was shown that any $n$-variable IP can be solved in $2^{O(n)} n^{2.5n}$
time (with polynomial dependencies on the remaining parameters).  Since then, progress on IP has been slow, though recent complexity improvements have
been made: the dependence on $n$ was reduced to $n^{2n}$~\cite{arxiv/HildebrandK10}, $\tilde{O}(n)^{\frac{4}{3}n}$~\cite{conf/focs/svp/DPV11}, and
finally $n^n$~\cite{thesis/D12}.

Let $K \subseteq \R^n$ denote a polytope. To find an integer point inside $K$, the general outline of the above algorithms is as follows. Pick a center
point $\vecc \in K$, and attempt to ``round'' $\vecc$ to a point in $\Z^n$ inside $K$. If this fails, decompose the integer program on $K$ into
subproblems. Here, the decomposition is generally achieved by partitioning $\Z^n$ along shifts of some rational linear subspace $H$ (often a hyperplane) and
recursing on the integral shifts of $H$ intersecting $K$.

In~\cite{conf/latin/cvp/D12}, an algorithm is given to perform the above rounding step in a ``near optimal'' manner. More precisely, the center
$\vecc$ of $K$ is chosen to be the center of gravity $\vecb$ of $K$ (which can be estimated via random sampling), and rounding $\vecb$ to $\Z^n$ is
done via an approximate CVP computation with target $\vecb$, lattice $\Z^n$, and norm $\|\cdot\|_{K-\vecb}$ (corresponding to scaling $K$ about
$\vecb(K)$). Here the AKS randomized sieve is used to perform the approximate CVP computation, which is efficient due to the fact that $K-\vecb$ is
near-symmetric (see \cite{MP00}). Let $\vecy \in \Z^n$ be the returned $(1+\eps)$-CVP solution, and assume that $\vecy$ is correctly computed (which
occurs with high probability). We can now examine the following cases. If $\vecy \in K$, we have solved the IP. If $\|\vecy-\vecb\|_{K-\vecb} >
(1+\eps)$, then by the guarantee on $\vecy$, for any $\vecz \in \Z^n$ we have that $\|\vecz-\vecb\|_{K-\vecb} > 1 \Leftrightarrow \vecz \notin K$.
Hence, we can immediately decide that $K \cap \Z^n = \emptyset$. Lastly, if $1 < \|\vecy-\vecb\|_{K-\vecb} \leq (1+\eps)$, we know that
$\frac{1}{1+\eps}K + \frac{\eps}{1+\eps}\vecb$ is integer free while $(1+\eps) K-\eps\vecb$ contains $\vecy$. In this final case, we are in
essentially a near-optimal situation for computing a ``good'' decomposition (using the so-called ``flatness'' theorems in the geometry of numbers).
We note with previous methods (i.e.~using only symmetric norm or $\ell_2$ techniques), the ratio of scalings between the integer free and non integer
free case was $O(n)$ in the worst case as opposed to $(1+\eps)^2$ (here $\eps$ can be any constant $\leq 1$). 

With the techniques in this paper, we note that the above rounding procedure can be made Las Vegas (i.e.~no probability of error, randomized running
time) by replacing the AKS Sieve with our new DPV based solver (randomness is still needed to estimate the center of gravity). This removes any
probability of error in the above inferences, making the above rounding algorithm easier to apply in the IP setting. We note that the geometry induced
by the above rounding procedure is currently poorly understood, and very little of it is being exploited by IP algorithms. One hope for improving the
complexity of IP with the above methods, is that with a strong rounding procedure as above one maybe able to avoid the worst case bounds on the number
of subproblems created at every recursion node. Currently, the main way to show that $K$ admits a small decomposition into subproblems is to show that
the covering radius of $K$ (i.e.~the minimum scaling such that every shift of $K$ intersects $\Z^n$) is large. Using the above techniques, we easily
get that in the final case the covering radius is $\geq \frac{1}{1+\eps}$ (since $\frac{1}{1+\eps}K + \frac{\eps}{1+\eps}\vecb$ is integer free),
however in reality the covering radius could be much larger (yielding smaller decompositions). Here, an interesting direction would be to try and show
that on the aggregate (over all subproblems), the covering radii of the nodes must grow as we go down the recursion tree. This would allow us to show
that as we descend the recursion tree, the branching factor shrinks quickly, allowing us to get better bounds on the size of the recursion tree (which
yields the dominant complexity term for current IP algorithms).

\paragraph{CVP under $\ell_\infty$.} While the ideas presented here do not seem to be practically implementable in general (at least
currently), there are special cases where the overhead incurred by our approach maybe acceptable. One potential target is solving
$(1+\eps)$-CVP under $\ell_\infty$. This is one of the most useful norms that is often approximated by $\ell_2$ for lack of a better
alternative. 

As an example, in~\cite{conf/arith/BrisebarreC07}, they reduce the problem of computing machine efficient polynomial approximations
(i.e.~having small coefficient sizes) of $1$ dimensional functions to CVP under $\ell_\infty$. The goal in this setting is to generate a
high quality approximation that is suitable for hardware implementation or for use in a software library, and hence spending considerable
computational resources to generate it is justified.  

We now explain why the $\ell_\infty$ norm version of our algorithms maybe suitable for practical implementation (or at least efficient
``heuristic'' implementation). Most importantly, for $\ell_\infty$ the DPV lattice point enumerator is trivial to implement. In
particular, to enumerate the lattice points in a cube, one simply enumerates the points in the outer containing ball and retains those in
the cube. Second, if one is comfortable with randomization, the sparsifier can be constructed by adding a simple random modular form to the
base lattice. For provable guarantees, the main issue is that the modulus must be carefully chosen (see Section~\ref{sec:lat-spar-rand}),
however it seems plausible that in practice an appropriate modulus may be guessed heuristically.

\bibliographystyle{alphaabbrvprelim}
\bibliography{lattices,crypto,cg,acg,extras}

\appendix

\section{Covering Bound}
\label{sec:appendix-1}

In this section, we prove the basic covering bound stated in Lemma \ref{lem:gkl-bnds}.

%
%


For a set $A \subseteq \R^n$, let ${\rm int}(A)$ denote the interior of $A$. For convex bodies $A,B \subseteq \R^n$, we define the
covering number $N(A,B) = \inf \set{|\Lambda|: \Lambda \subseteq \R^n, A \subseteq \Lambda + B}$, i.e.~the minimum number of translates of $B$ needed to cover $A$. We will require the following standard inequality on the covering number.
 
\begin{lemma}
\label{lem:cov-est}
Let $A,B  \subseteq \R^n$ be convex bodies, where $B$ is symmetric. Then
\[
   N(A,B) \leq \frac{\vol_n(A + B/2)}{\vol_n(B/2)}.
\]
\end{lemma}
\begin{proof}
Let $T \subseteq A$ be any maximal set of points such that for all distinct $\vecx,\vecy \in T$, $(\vecx+B/2) \cap (\vecy+B/2) = \emptyset$.
We claim that $A \subseteq T + B$. For any $\vecz \in A$, note by maximality of $T$ that there exists $\vecx \in T$ such that $(\vecz+B/2) \cap (\vecx+B/2) \neq \emptyset$. 
Therefore $\vecz \in \vecx + B/2 - B/2 = \vecx + B$, as needed. 

Since $T+B/2$ corresponds to $|T|$ disjoint translates of $B/2$, we have that
\[
|T| \vol_n(B/2) = \vol_n(T+B/2) \leq \vol_n(A+B/2) \text{.}
\]
Rearranging the above inequality yields the lemma.
\end{proof}

\begin{proof}[Proof of Lemma \ref{lem:gkl-bnds}]
We prove the bound on $G(dK,\lat)$ in terms of $\lambda_1(K \cap -K,\lat)$. \newline Let $s = \frac{1}{2} \lambda_1(K \cap -K,\lat)$. For $\vecx \in \lat$, we examine
\[ \vecx + \mathrm{int}(s(K \cap -K)) = \set{\vecz \in \R^n: \length{\vecz-\vecx}_{K
\cap -K} < s}.\] Now for $\vecx, \vecy \in \lat$, $\vecx \neq \vecy$, we claim that
\begin{equation}
\vecx + \mathrm{int}(s(K \cap -K)) \cap \vecy + \mathrm{int}(s(K \cap -K)) = \emptyset
\label{eq:no-int-1}
\end{equation}
Assume not, then $\exists ~ \vecz \in \R^n$ such that
$\length{\vecz-\vecx}_{ K\cap-K },\length{\vecz-\vecy}_{ K \cap -K} < s$. Since $K \cap -K$ is symmetric, we note
that $\length{\vecy-\vecz}_{K \cap -K}  = \length{\vecz-\vecy}_{K \cap -K} < s$. But then we have that
\begin{align*}
\length{\vecy-\vecx}_{K \cap -K} &= \length{\vecy-\vecz+\vecz-\vecx}_{K \cap -K} \leq \length{\vecy-\vecz}_{K \cap -K} + \length{\vecz-\vecx}_{K \cap -K} \\
&< s + s = 2s = \lambda_1(K \cap -K,\lat),
\end{align*}
a clear contradiction since $\vecy-\vecx \neq 0$.

Take $\vecc \in \R^n$. To bound $G(dK,\lat)$ we must bound $|(\vecc + dK) \cap \lat|$. For $\vecx \in \vecc + dK$, we note that $\vecx + s(K \cap
-K) \subseteq \vecc + (d+s)K$. Therefore,
\begin{align*}
\vol_n((d+s) K) &= \vol_n(\vecc + (d+s)K) \geq \vol_n\left( ((\vecc + dK) \cap \lat) + s(K \cap -K) \right) \\
	  &= |(\vecc + dK) \cap \lat| \vol_n(s(K \cap -K))
\end{align*}
where the last equality follows from $(\ref{eq:no-int-1})$. Therefore, we have that
\[
|(\vecc + dK) \cap \lat| \leq \frac{\vol_n((d+s)K)}{\vol_n(s(K \cap -K))} = \left(\frac{d+s}{\gamma s}\right)^n = \gamma^{-n} \left(1 +
\frac{2d}{\lambda_1(K \cap -K,\lat)}\right)^n
\]
as needed.

We prove the bound on $G(dK,\lat)$ in terms of $|(K \cap -K) \cap \lat|$.  Examine $dK+\vecx$. Let $\vecy_1,\dots,\vecy_N \in (tK+\vecx) \cap \lat$, denote a maximal
collection of points such that the translates $\vecy_i + \frac{1}{2}(K \cap -K)$, $i \in [N]$, are interior disjoint. We claim
that $(dK+\vecx) \cap \lat \subseteq \cup_{i=1}^N \vecy_i + (K \cap -K)$. Take $\vecz \in (dK+\vecx) \cap \lat$. Then by construction
of $\vecy_1,\dots,\vecy_N$, there exists $i \in [N]$ such that 
\[
\vecz + \frac{1}{2}(K \cap -K) \cap \vecy_i + \frac{1}{2}(K \cap -K) \neq \emptyset \Rightarrow \vecz \in \vecy_i + (K \cap -K)
\]
as needed. Therefore $|(dK+\vecx) \cap \lat| \leq \sum_{i=1}^n |(\vecy_i + (K \cap -K)) \cap \lat| = N |(K \cap -K) \cap \lat|$.
Since $K$ is $\gamma$-symmetric, we get that
\[
N = \frac{\vol_n(\cup_{i=1}^n \vecy_i + \frac{1}{2}(K \cap -K))}{\vol_n(\frac{1}{2}(K \cap -K))} 
    \leq 2^n\gamma^{-n} \frac{\vol_n(dK + \frac{1}{2}(K \cap -K))}{\vol_n(K)} \leq \gamma^{-n} (2d+1)^n
\]
as needed. Since the above bound holds for all $\vecx \in \R^n$, we get that \newline $G(tK,\lat) \leq \gamma^{-n} (2d+1)^n \cdot |(K \cap -K) \cap \lat|$ as needed.
\end{proof}

\end{document}